\documentclass[aps,twocolumn,showpacs,amsmath,amssymb, superscriptaddress]{revtex4}

\usepackage{subfigure}
\usepackage{amsmath}
\usepackage{amssymb}
\usepackage{amsfonts}

\usepackage[dvips]{graphicx}
\usepackage{float}
\usepackage{amsthm}

\newtheorem{theorem}     {Theorem}
\newtheorem*{theorem*}     {Theorem}

\newtheorem*{proposition*} {Proposition}

\newtheorem*{corollary*}   {Corollary}

\theoremstyle{definition}

\newtheorem*{conjecture*}  {Conjecture}

\newtheorem*{remark*} {Remark}

\usepackage{amsthm}

\DeclareMathOperator{\Tr}{Tr}

\DeclareMathOperator{\U}{\rm U}
\DeclareMathOperator{\SU}{\rm SU}
\DeclareMathOperator{\su}{\mathfrak s\mathfrak u}

\newcommand{\paren}[1]{{\left( #1 \right)}}
\newcommand{\bigparen}[1]{\bigl({#1}\bigr)}

\newcommand{\bracket}[1]{{\left[ #1 \right]}}

\newcommand{\floor}[1]{{\lfloor #1 \rfloor}}

\renewcommand{\Re}{{\rm Re}}

\makeatletter
\def\loweq@align#1#2{\lower.6ex\vbox{\baselineskip\z@skip\lineskip\z@
    \ialign{$\m@th#1\hfil##\hfil$\crcr#2\crcr=\crcr}}}
\def\lowsim@align#1#2{\lower.6ex\vbox{\baselineskip\z@skip\lineskip\z@
    \ialign{$\m@th#1\hfil##\hfil$\crcr#2\crcr\sim\crcr}}}
\def\geqq{\mathrel{\mathpalette\loweq@align >}}
\def\leqq{\mathrel{\mathpalette\loweq@align <}}
\def\grsim{\mathrel{\mathpalette\lowsim@align >}}
\def\gsim{\mathrel{\mathpalette\lowsim@align >}}
\def\lsim{\mathrel{\mathpalette\lowsim@align <}}
\makeatother
\newcommand{\grless} 
{ {\, \raise-.24em\hbox{$<$} \hspace{-0.8em} \raise.31em\hbox{$>$}\, } }
\newcommand{\lessgr} 
{ {\, \raise-.24em\hbox{$>$} \hspace{-0.8em} \raise.31em\hbox{$<$}\, } }
\newfont{\bg}{cmr10 scaled\magstep4}                    
\newcommand{\bigzerou}{\smash{\lower1.7ex\hbox{\bg 0}}}

\newcommand{\nn}{\nonumber \\ }

\newcommand{\C}{{\mathbb C}}

\newcommand{\FF}{{\cal F}}

\newcommand{\TR}[1]{{#1}_{\rm rev}}
\newcommand{\TRB}{_{\rm rev}}

\newcommand{\crl}[1]{[-\infty,\infty]}

\newcommand{\ket}[1]{|{#1}\rangle}

\newcommand{\Ref}[1]{(\ref{#1})}

\newcommand{\wb}{\overline}

\newcommand{\da}[1]{#1^\dag}

\newcommand{\pa}[1]{\sigma_{#1}}

\newcommand{\sa}{self-adjoint}

\newcommand{\alg}[1]{{\mathfrak #1}}

\newcommand{\g}{{\alg g}}

\newcommand{\old}{_{\rm in}}
\newcommand{\new}{_{\rm out}}
\newcommand{\qft}{_{\rm QFT}}

\newcommand{\newup}{^{\rm out}}
\newcommand{\qct}{QCT}
\newcommand{\tqct}{{\it t}\/-QCT}
\newcommand{\fqct}{{\it f}\/-QCT}
\newcommand{\gqct}{{\it g}\/-QCT}
\newcommand{\NOT}{{\sc NOT}}
\newcommand{\SWAP}{{\sc SWAP}}
\newcommand{\nos}{_{\scriptsize\rm asym}}

\newcommand{\comment}{}
\newcommand{\commenton}{
  \renewcommand{\comment}[1]{}
}

\newcommand{\upconv}{\to}
\newcommand{\id}{{\bf 1}}
\newcommand{\Tmax}{T_{2,\scriptsize\mbox{max}}}
\newcommand{\Lmax}{L_{2,\scriptsize\mbox{max}}}
\newcommand{\ov}[1]{$\overline{#1}$}
\newcommand{\psp}[1]{^{(#1)}}
\newcommand{\prepri}{}

\commenton

\begin{document}

\title{%
Time complexity and gate complexity
}
\author{Tatsuhiko Koike}
\email{koike@phys.keio.ac.jp}
\affiliation{Department of Physics, Keio University, Yokohama,
  Japan}
\author{Yosuke Okudaira}
\affiliation{Department of Physics, Tokyo Institute of Technology, 
  Tokyo, Japan
}

\date{\today}


\begin{abstract}
We formulate and investigate the simplest version of 
time-optimal quantum computation theory ({\tqct}), 
where 
the computation time 
is defined by 
the physical one 
and 
the Hamiltonian contains only one- and two-qubit interactions. 
This version of {\tqct} 
is also considered as optimality by sub-Riemannian geodesic length. 
The work has two aims: 
one is to develop a {\tqct} itself based on
physically natural concept of time, and 
the other is to pursue the possibility of using {\tqct} as a 
tool to estimate the complexity in conventional gate-optimal quantum
computation theory (\gqct). 
In particular, 
we investigate to what extent is true the statement: 
time complexity is polynomial in the number of qubits if and only if 
so is gate complexity. 
In the analysis, 
we relate {\tqct} and optimal control theory (OCT) 
through fidelity-optimal computation theory ({\fqct}); 
{\fqct} is equivalent to {\tqct} 
in the limit of unit optimal fidelity, 
while it is formally similar to OCT. 
We then 
develop an efficient numerical scheme for {\fqct} by 
modifying Krotov's method in OCT, 
which has monotonic convergence property. 
We implemented the scheme 
and 
obtained solutions of {\fqct} and of {\tqct} 
for the quantum Fourier transform and a unitary operator that does not
have an apparent symmetry. 
The former has a polynomial gate complexity and 
the latter is expected to have exponential one which is based on the
fact that a series of generic unitary operators has a exponential gate
complexity. 
The time complexity for the former is found to be linear
in the number of qubits, which is understood naturally by 
the existence of an upper bound. 
The time complexity for the latter is exponential in 
the number of qubits. 
Thus the both targets seem to be examples satisfyng the statement above. 
The typical characteristics of the optimal Hamiltonians are 
symmetry under time-reversal and constancy of one-qubit operation, 
which are mathematically shown to hold in fairly general situations. 
\end{abstract}
\maketitle

\section{Introduction}
Quantum computation is performed by physical processes obeying quantum mechanics. 
It became one of the most exciting field in physics and information theory 
after Shor~\cite{shor} 
discovered an algorithm to factorize integers 
which is 
exponentially faster than any known classical ones. 
In quantum computation theory (QCT), 
as in classical computation
theory, the computation time is usually defined by the number of
elementary steps or gates necessary to perform a computation, i.e., 
to realize a
desired unitary operator. 
Minimum such number 
is called the gate
complexity. 
We shall call this conventional QCT as gate-optimal QCT (\gqct). 

In this paper, we investigate time-optimal quantum computation theory 
({\tqct}) 
where the computation time is defined by {\em physical time}. 
There are 
two motivations for {\tqct}. 
The first is to develop 
(an abstract) 
{\tqct} itself which is a physical-time-based 
alternative to {\gqct}. 
Since a quantum computation is a physical process, 
it is physically more natural to measure the time by 
the {physical one}. 
From this viewpoint, the computation time in 
{\gqct} can be seen as 
{\em information-theoretic time}
which is a more abstract or coarse-grained 
notion of time than the physical one. 
Time optimality 
is attracting growing attention in quantum optimal control 
theory (OCT) 
mainly in the context of physical applications 
such as control of an atom by an electromagnetic field 
and NMR quantum computation 
(\cite{khaneja,schulte} and references therein). 

The second motivation, which we stress in the present work 
and state in detail in Sec.~\ref{mot}, 
is that {\tqct} may be a useful tool to analyze 
{\gqct}. 
Finding the gate-optimal algorithm 
is a discrete and combinatorial problem, 
which makes construction of a general theory  difficult. 
On the other hand, 
time-optimal algorithms are smooth curves in a certain space
which obeys a differential equation~\cite{paper1,paper2,paper3},  
typically that for 
a sub-Riemannian geodesic~\cite{montgomery} on a manifold. 
This may allow a general theory 
and approximation methods. 
Moreover, 
roughly speaking, 
upper and lower bounds for gate complexity 
can be given in terms of 
optimal physical time~\cite{nielsen1,nielsen2}. 
Thus {\tqct} is useful in the investigation of {\gqct}. 
One may become able to calculate gate complexity 
by calculating time complexity. 
We ask to what extent holds the statement that the time complexity is 
polynomial in the number of qubits if and only if so is the gate complexity. 

We will therefore compare the time complexity and the gate complexity 
for some typical examples. 
We choose 
the 
quantum Fourier transform (QFT) 
as an example of the target unitary operator 
of which a fast algorithm in the sense of {\gqct} 
(i.e. whose gate complexity is polynomial in 
the number of qubits) is known, 
while we choose 
a target unitary with no special symmetry
because a generic series of target unitary operators 
is known not to have fast algorithms. 

To achieve it, 
we will make use of an efficient numerical method for OCT, 
so-called Krotov's method. 
We relate {\tqct} to fidelity-optimal {\qct} (\fqct), 
and develop a Krotov-like scheme for {\fqct} by making use of 
the formal similarity of {\fqct} and OCT. 
In the context of OCT, similar ideas of replacing time optimality to
fidelity optimality have been used. 

We will see that both the QFT and 
the asymmetric unitary operator 
satisfy the statement above. 
Furthermore, we will find some 
characteristic behaviour of the optimal Hamiltonian, 
namely, time-reversal symmetry and constancy of one-qubit 
Hamiltonian components. 

In analyzing {\tqct}, it is useful 
to combine numerical and mathematical approaches. 
We will show mathematically 
that the behavior which is found numerically 
is satisfied in fairly general situations. 
These arguments in turn support the soundness of our numerical calculation. 

The organization of the paper is as follows. 
In Sec.~\ref{tqct}, we introduce {\tqct} as a special case of 
quantum brachistochrone~\cite{paper1,paper2,paper3} 
and explain our above-mentioned motivation more precisely. 
In Sec.~\ref{fqct}, we discuss the relation between {\tqct} and {\fqct}. 
In Sec.~\ref{krotov}, we present a Krotov-like numerical method for 
{\fqct}. 
We will show our numerical results and extract the properties of the 
solutions of {\fqct} and of {\tqct} in Sec.~\ref{res}. 
We will give a proof of time-reversal invariance and constancy of one-qubit
components in Sec.~\Ref{math}. 
Sec.~\ref{conc} is for conlusion. 
In 
Appendix~\ref{krotov-mono} we show the monotonicity of the numerical 
scheme of Sec.~\ref{krotov}, 
and in Appendix~\ref{pf-time-reversal} we give a proof of a theorem in
Sec.~\ref{math}. 

We use the units $\hbar=1$.

\section{Time-optimal QCT}
\label{tqct}
In this section, we introduce {\tqct}, 
with some review of quantum brachistochrone, 
and present the motivation of the paper
in more detail. 

\subsection{Definition}
\label{tqct-def}
Let us define the simplest version of
{\em time-optimal 
{\qct} ({\tqct})}\/ 
as a special case of quantum brachistochrone for unitary 
operations~\cite{paper2}, namely, the case 
in which the Hamiltonian $H(t)$ 
involves only one- and two-qubit interactions 
and is subject to a normalization constraint. 
Below is a summary of the formalism~\cite{paper2} in this case. 

Quantum brachistochrone for unitary operations is a framework to find the 
optimal Hamiltonian $H(t)\in\Gamma$ 
which realizes the desired unitary 
$U_f$ up to phase in the minimum time $T\geqslant0$, 
where 
$\Gamma$ is the set of available
Hamiltonians. 
Namely, one wants to find the minimum $T$ such that 
there is 
a unitary operator $U(t)$ satisfying the Schr\"odinger equation
\begin{align}
  i\dot U(t)=H(t)U(t), 
  \label{eq-schu} 
  \end{align}
and the initial and final conditions, 
  \begin{align}
U(0)&={\id}, 
\label{eq-iniu}
\\
U(T)&=e^{-i\chi} U_f, 
\label{eq-finu}
\end{align}
where $\id$ is the identity operator and $\chi$ is some real. 

For {\tqct} of the system of $n$ qubits, 
the set $\Gamma$ of available Hamiltonians consists of {\sa} operators 
\begin{align}
  H(t)= 
  \sum_{a}h_{a}(t)\tau_{a}
  \label{eq-Gamma}
\end{align}
with a normalization constraint
\begin{align}
  |h(t)|^2:=
  \sum_{a}(h_{a}(t))^2=N\omega^2, 
  \label{eq-normalization}
\end{align}
where $N:=2^n$ and 
the basis $\{\tau_a\}$ 
consists of 
$\sigma^{a}_{j}/\sqrt N$ and $\sigma^{ab}_{jk}/\sqrt N$. 
Here, 
$\sigma^{a_1\cdots a_m}_{j_1\cdots j_m}$
($1\leqslant a_1<\cdots<a_m\leqslant n$ and $j_l=x,y,z$)
denotes 
the direct 
product of 
the Pauli operator 
$\pa {j_l}$ on the $a_l$th qubit 
and identities on the others; 
for example, $\sigma^{13}_{xy}=\pa x\otimes1\otimes\pa
y\otimes1\otimes\cdots\otimes1$. 
The normalization condition \Ref{eq-normalization} 
can be interpreted physically 
as finiteness of available energy in operations, 
while it is needed mathematically 
for the optimality problem to be 
well-posed~\cite{paper1}. 
The $N$-dependence of 
\Ref{eq-normalization} 
is for 
consistency under composition of systems%
~\cite{paper3quantph}. 
The parameter $\omega$ can be interpreted as defining a unit for $T$. 
The problem is a particular case of linear homogeneous 
constraints~\cite[Sec. III]{paper2}. 

This is the natural counterpart in {\tqct} 
to the standard paradigm in {\gqct} 
where one constructs a desired unitary by a
sequence of one- and two-qubit operations. 
The parameter $\omega$ can be interpreted as defining a unit for $T$. 
The problem is a particular case of quantum
brachistochrone for linear homogeneous 
constraints~\cite[Sec. III]{paper2}. 

In the form of variational principle, {\tqct} is 
to minimize 
the action 
\begin{align}
  S(U,h,V,\wb\lambda)
  =
  \int_0^T dt \Bigl[L_T+L_S
  +\frac{\check\lambda(t)}2(|h(t)|^2-N\omega^2)\Bigr], 
  \label{eq-action-brach}
\end{align}
where 
$L_T:=
  \sqrt
  {\frac{\Tr{\dot U\da{} (1-P_U) (\dot U)}} 
    {\Tr{(HU)\da{}(1-P_U)(HU)}}}$, 
$P_U(A):=\frac1N(\Tr A\da U)U$, 
$L_S:=2\Re\Tr\da V\bigparen{i\dot U-HU}$, 
an overdot denotes time derivative. 
The first term counts the time duration, 
where $L_T$ 
is unity when the Schr\"odinger equation~\Ref{eq-schu} 
holds and 
is invariant under time reparametrization 
$t\mapsto f(t)$~\cite{fn-a}. 
The second term $L_S$ guarantees that 
the Schr\"odinger equation \Ref{eq-schu} holds at all times where
the unitary operator $V(t)$ is the Lagrange multipliers. 
The third term guarantees 
the normalization constraint \Ref{eq-normalization} 
where the real function 
$\check\lambda(t)$ is Lagrange multiplier. 
We have adopted an action equivalent to 
but slightly different in form from that in Ref.~\cite{paper2} 
for better connection with the arguments below.

We note that the phase of $U(t)$ does not matter in the present formulation of
{\tqct}. In fact, the action~\Ref{eq-action} is invariant under a 
time-dependent phase change of $U(t)$ 
which can be considered as a gauge transformation~\cite{paper2}. 
Therefore the theory is defined on $\U(N)/\U(1)$, and 
one can also think of it as a theory on $\SU(N)$ (by ``gauge fixing''). 

The Euler-Lagrange equations are~\cite{paper2}
the Schr\"odinger equation \Ref{eq-schu} for 
$U(t)$, 
the normalization constraint \Ref{eq-normalization}, 
the Schr\"odinger equation 
for $V(t)$, 
\begin{align}
  i\dot V(t)&=H(t)V(t), \label{eq-schv}
\end{align}
and the equation determining $H(t)$, 
\begin{align}
  \lambda(t) h_a(t)&=\Tr \tau_aF(t),
  \quad F(t):=U(t)\da V(t)+V(t)\da U(t), 
  \label{eq-lambda}
\end{align}
where $\lambda(t):=\check\lambda(t)-\frac1{N\omega^2}$. 
One must solve these equations 
with the initial and final conditions 
\Ref{eq-iniu} and 
\Ref{eq-finu}. 

Let us recall some general features of the system. 
First, 
$F$ satisfies the simple evolution equation 
\begin{align}
i\dot F(t)=[H(t),F(t)], 
  \label{eq-brach}
\end{align}
which follows from 
\Ref{eq-schu}, \Ref{eq-schv} and \Ref{eq-lambda}. 
Second, 
the expression for $\lambda(t)$ is 
\begin{align}
  \lambda(t)=
  \frac1{\omega}\sqrt{\frac1{N}\sum_a\paren{\Tr \tau_aF(t)}^2}, 
  \label{eq-lambda-2}
\end{align}
which follows from 
\Ref{eq-normalization} and \Ref{eq-lambda}. 
Third, this variable $\lambda(t)$ is constant in 
time~\cite{fn-b}. 

In the special case when $U_f$ is 
a one- or two-qubit operation in $\U(N)$, 
the solution of {\tqct} 
is given by a Riemannian geodesic 
$U(t)=e^{-iHt}$  
on $\U(N)/\U(1)$, where 
$H$ is constant~\cite{paper1,paper2}. 
The time $T(U_f)$ 
is proportional to the arc length 
and depends solely on the eigenvalues $e^{i\theta_j}$ of 
the relevant $\SU(4)$ part of $U_f$ as
\begin{align}
T(U_f)=\frac1{2\omega} \sqrt{\min_{\chi,m_j}
\sum_{j=1}^4(\theta_j+2\pi m_j-\chi)^2}, 
\label{eq-T-for-2bit}
\end{align}
where $\chi$ is a real and $m_j$ are integers. 
For example, the time is given by 
$T(U_{\scriptsize\mbox{CNOT}})=
T(S)=\frac{\sqrt3\pi}{4\omega }$ for the 
controlled {\NOT} gate $U_{\scriptsize\mbox{CNOT}}$ 
and {\SWAP} gate $S$ whose eigenvalues are $(1,1,1,-1)$, 
and 
$\Tmax=\frac{\sqrt5\pi}{4\omega}$ for 
the hardest two-qubit operation whose eigenvalues are $(1,i,-1,-i)$. 
These can be used for the unit of $T(U_f)$ 
in comparing $T(U_f)$ with $G(U_f)$.

\subsection{{\tqct}, {\gqct} and our motivation} 
\label{mot}

There are some rigorous relations between 
the gate complexity and time complexity. 
Very roughly speaking, one can give 
upper and lower bounds for gate complexity through 
the time complexity. 

Before introducing the relations, recall that 
in the simple cases, the time complexity $T(U)$ 
can also be interpreted as the arc length $L(U)$ of the sub-Riemannian 
geodesic connecting $\id$ and $U_f$, 
up to overall multiplicative constant. 
The simplest version of {\tqct} presented in Sec.~\ref{tqct-def} 
falls into this 
category, so that $\frac{T(U)}{\Tmax}$ appearing below in this paper 
can also be considered as 
$\frac{L(U)}{\Lmax}$ where $\Lmax$ is the sub-Riemannian geodesic distance 
between $\id$ and the furthest two-qubit operation. 

The precise relations are given~\cite[esp. Eq. (15)]{nielsen2} by
\begin{align}
  \frac{T(U)}{\Tmax}\leqslant G(U),
  \quad 
  G(U,\epsilon)\leqslant \frac{AT(U)^3n^6}{\epsilon^2}, 
  \label{eq-T-vs-G}
\end{align}
where $T(U)$ is the time complexity to realize
the unitary $U$, and $G(U,\epsilon)$ is the gate complexity to realize
a unitary 
within the distance $\epsilon$ from $U$ (measured in the
operator norm), 
$\Tmax$ is the constant defined in the previous section, 
and $A$ is some constant. 
Note that we always have 
$G(U,\epsilon)\leqslant G(U)$. 

Let us explain the first inequality of \Ref{eq-T-vs-G} 
which is simple. 
For a one- or two-qubit gate $U$, 
we have $\frac{T(U)}{\Tmax}\leqslant 1=G(U)$ by the definition of $\Tmax$. 
Then, for a general $U$, letting $U=U_m\cdots U_1$ be the gate-optimal
decomposition, we have $\frac1{\Tmax}\sum_{j=1}^m T(U_j)\leqslant G(U)$. 
However, the sum in the LHS, the time cost of this decomposition, 
must be greater than or equal to the time complexity $T(U)$ of $U$. 
Thus the first inequality of \Ref{eq-T-vs-G} holds.

The relation \Ref{eq-T-vs-G} is suggestive, and 
one might expect that 
\begin{align}
  G(U)\approx T(U), 
  \label{eq-G-T}
\end{align}
by which we mean that 
$G(U)$ and $T(U)$ are bounded by some polynomial of each other 
and $n$. 
However, it is argued that this cannot be true in general~\cite{nielsen2}. 
Then, 
one may want to ask to what extent \Ref{eq-G-T} is true in general, 
since 
\Ref{eq-T-vs-G} is derived by a completely general argument. 
In particular, we think that it is interesting to ask
what are the classes of unitary operators $U$ which 
satisfy 
\begin{enumerate}
\item 
$T(U)\approx G(U)$, i.e., 
$T(U)$ and $G(U)$ 
are bounded by a polynomial of each other and $n$, 
\item 
$T(U)\sim G(U)$, i.e., 
$T(U)$ is polynomial in $n$ if and only if so is $G(U)$. 
\end{enumerate}
The second class contains the first class. 
Note that the first and second conditions above are equiavlent to 
$G(U)\approx G(U,\epsilon)$ and 
$G(U)\sim G(U,\epsilon)$, respectively, 
by \Ref{eq-T-vs-G}. 
The questions above are not easy, but 
we want to develop a basis here 
which will help answering these questions in the future.

{\tqct} itself is a good framework to analyze these questions theoretically. 
However, a drawback is that 
it is extremely hard 
in practice 
to derive exact solutions 
when $n$ is large. 
One way is to appeal to numerical calculations, 
but it is also difficult 
because one encounters a two-point boundary value problem 
in dimensions rapidly increasing with $n$. 
In Ref.~\cite{schulte}, 
time-optimal solutions were obtained up to seven qubits
for the like of Hamiltonians 
appearing 
in 
NMR 
quantum computers. 
That problem involves $2n$ ($\lesssim 20$ for $n\lesssim10$) 
functions for the boundary value problem. 
Our {\tqct} has $9n(n-1)/2+3n$ functions, 
where the number easily becomes several hundred. 
To minimize the numerical difficulty, we 
introduce 
a problem 
bridging {\tqct} and optimal control theory (OCT), 
and make use of an efficient numerical scheme for the latter.

\section{Fidelity-optimal QCT}
\label{fqct}
In this section, we relate {\tqct} to {\fqct} and prepare for introducing the
numerical method in the next section. 
In the context of OCT, 
similar ideas of relating time optimality and fidelity 
optimality are used (e.g. \cite{schulte}). 

\subsection{Definition}
\label{fqct-def}
Let us define {\em fidelity-optimal {\qct} ({\fqct})}\/ as a framework
to solve the following
problem: 
given a target unitary operator $U_f$ and time interval $T$, 
find the Hamiltonian $H(t)\in \Gamma$ as a function of time which maximizes 
the fidelity $\FF(U(T),U_f)$. 
Here we use the trace fidelity 
\begin{align}
  \FF(U(T),U_f):=\frac1N|\Tr\da U U_f|
\end{align}
because we want $U(T)$ to be close to $U_f$ only up to phase. 
The problem is to minimize the action 
\begin{align}
  S(U,h,V,\lambda)
  &=
  -N\FF(U(T),U_f)^2
  \nn
  &+ \int_0^T dt 
  \Bigl[L_S+\frac{\lambda(t)}2(|h(t)|^2-N\omega^2)\Bigr]. 
  \label{eq-action}
\end{align}
The first term is the squared fidelity of the final unitary $U(T)$ with respect
to the target $U_f$. 
The second and third terms guarantee the Schr\"odinger equation \Ref{eq-schu}
and the normalization constraint \Ref{eq-normalization}, respectively, 
as in the case of \tqct. 
The action \Ref{eq-action} is invariant under the (time-dependent) 
change of the  phase of $U(t)$ so that 
the theory is defined on $\U(N)/\U(1)$. 

The Euler-Lagrange equations yield
\Ref{eq-schu}, \Ref{eq-schv}, 
\Ref{eq-lambda}, and 
\begin{align}
  V(T)&=\frac iN U_f{\Tr U_f^\dagger U(T)}. 
  \label{eq-finv}
\end{align}
Therefore all the equations \Ref{eq-schu}--\Ref{eq-lambda-2} hold except
that 
\Ref{eq-finu} is replaced by \Ref{eq-finv}. 
As in the case of {\tqct}, 
$\lambda(t)$ is constant in time. 

The solution 
gives the maximal fidelity $\FF(U(T),U_f)$ for given $T$. 
{\fqct} is important on its own when one discusses the
tradeoff between speed and error of computation. 
here we use {\fqct} to bridge {\tqct}
and OCT.

\subsection{Relation between solutions to {\fqct} and {\tqct}}
\label{fqct-vs-tqct}
The solution of {\tqct} can be obtained from that of 
{\fqct} in the limit $\FF\to1$. 

First, {\fqct} is equivalent to 
the problem of 
minimizing the physical time to achieve given fidelity. 
Suppose the solution $U_0(t)$ of {\fqct} for given time $T_0$, with 
optimal fidelity $\FF(U(T_0),U_f)=f$, 
was not the solution of the above new problem. 
There would be $\check U(t)$ achieving the same fidelity in some $T<T_0$. 
Then one could construct $U(t)$ with 
$\FF(U(T_0),U_f)>f$ by 
defining $U(t)=\check U(t)$ for $0\leqslant t\leqslant T$ and 
appropriately in $T< t\leqslant T_0$. 
This would contradict to the fact that 
$U_0(t)$ is a solution of {\fqct}. 
Thus $U_0(t)$ is the solution of the above problem. 
The converse is shown similarly. 

Second, the above new problem with fixed fidelity $f$ 
yields {\tqct} when $f\to1$. 
Therefore the solution of {\fqct} 
gives that of {\tqct} 
in the limit 
$\FF\upconv1$. 

\subsection{Formal similarity of {\fqct} and OCT}
\label{fqct-vs-oct}
Let us see the formal similarity between {\fqct} and OCT. 

Our action \Ref{eq-action} defines an 
exact fidelity optimality problem with constraints. 
If we replace $\lambda(t)$ with a given constant or a given function, 
the term $\frac{\lambda(t) |h(t)|^2}2$ in the integral 
can be interpreted as a penalty term while 
the constant term $-\frac{\lambda(t) N\omega^2}2$ may be dropped. 
Then the problem is to 
minimize 
a 
combination of the fidelity and penalty terms with their weights 
specified by $\lambda(t)$, 
which is 
a typical problem in OCT. 
This is the formal relation between {\fqct} and OCT.

\section{A Krotov-like scheme}
\label{krotov}
In this section, 
\label{krotov-def}
we shall define an efficient numerical 
scheme for {\fqct} by 
modifying 
Krotov's method~\cite{tannor}
in 
OCT, 
making use of the similarity between {\fqct} and OCT 
in Sec.~\ref{fqct-vs-oct}.

In what follows, the functions 
$\wb h_a(t)$ and $\wb\lambda(t)$, respectively,  are 
$h_a(t)$ and $\lambda(t)$ calculated in the middle of 
an iteration cycle of the scheme, 
and 
Eqs. \ov{\Ref{eq-normalization}}, 
\ov{\Ref{eq-schv}} and 
\ov{\Ref{eq-lambda}}, respectively, 
denote
Eqs. \Ref{eq-normalization}, \Ref{eq-schv} and \Ref{eq-lambda} 
with 
$h_a(t)$ (including that in $H(t)$) and $\lambda(t)$ being replaced 
by $\wb h_a(t)$ and $\wb\lambda(t)$. 
The scheme is as follows. 

(i) Prepare a seed Hamiltonian components $h_a(t)$, $0\leqslant t\leqslant T$. 

(ii) Set $U(0)=1$ and evolve $U(t)$ from $t=0$ to $t=T$ by
\Ref{eq-schu} with the Hamiltonian $H(t)=\sum_ah_a(t)\tau_a$. 

(iii) 
Set $V(T)$ by \Ref{eq-finv} and 
evolve $(V(t), \wb h_a(t), \wb\lambda(t))$ 
backward in time from $t=T$ to $t=0$ by 
Eqs. 
\ov{\Ref{eq-schv}},  
\ov{\Ref{eq-lambda}} and 
\ov{\Ref{eq-lambda-2}}, 
while $U(t)$ is treated as a given function and is not updated. 

(iv) Set $U(0)=1$ and 
evolve $(U(t), h(t), \lambda(t))$ forward in time from $t=0$ to $t=T$ 
by Eqs. \Ref{eq-schu}, \Ref{eq-lambda} and \Ref{eq-lambda-2}, 
while $V(t)$ is treated as a given function. 

(v) Repeat steps (iii) and (iv) until the variables converge; 
the final $h_a(t)$ defines the optimal Hamiltonian and 
the final $\FF(U(T),U_f)$ gives the maximal achievable fidelity 
in $T$. 

Note that the multiplier 
$\lambda(t)$ converges to an unknown constant 
only after the convergence, which is in contrast to 
original Krotov's scheme 
where $\lambda$ is a given constant parameter of the problem. 
The constancy of $\lambda(t)$ can be 
used for a convergence check of the present
scheme.  
The present scheme is also different 
from Krotov's method for the case of a fixed reference 
energy (see, e.g. \cite{kosloff2}). 
In the former, 
the variable $\lambda(t)$ is a
Lagrange multiplier 
and the normalization \Ref{eq-normalization} must be
satisfied at all times, 
and $\lambda(t)$ is updated in the scheme. 
In the latter, 
$\lambda(t)$ is a given function
which determines the amount of penalty imposed on the error in 
\Ref{eq-normalization}, 
and $\lambda(t)$ is not updated in the scheme. 
See also discussions in Sec.~\ref{fqct-vs-tqct}.

An important property of our Krotov-like scheme is 
monotonicity, which is necessary for the scheme to be useful. 
We give the proof in Appendix~\ref{krotov-mono}.

\section{Results}
\label{res}
We implemented the numerical scheme presented in Sec.~\ref{krotov} 
and performed calculations  for {\fqct}. 
We chose two examples for the target unitary operator $U_f$. 

The first is the 
quantum Fourier transform (QFT) $U\psp n\qft$, defined by
\begin{align}
  U\psp n\qft\ket{x}=\frac{1}{\sqrt N}\sum_{k=0}^{k=N-1}e^{2\pi ikx/N}\ket{k}, 
  \quad 
  x=0,1,...,N-1, 
\end{align}
which has 
the gate complexity $G(U_f)$ polynomial in $n$. 
In fact, a simple efficient algorithm (e.g. \cite{chuangnielsen})
is given by 
\begin{align}
  U\psp n\qft&=S_{n_2,n-n_2}\cdots S_{1n} U_1 \cdots U_n, 
  \nn
  U_j&:=R_{j,n-j+1,n}\cdots R_{2,n-j+1,n-j+2}W_{n-j+1}, 
  \label{eq-fourier-algorithm}
\end{align}
where 
$W_j$ is the Walsch-Hadamard gate 
$W:=
\frac1{\sqrt2}
\begin{bmatrix}
  1&1\\1&-1
\end{bmatrix}
$ applied on the $j$th qubit, 
$R_{j,km}$ is the $2\pi/2^j$-phase shift gate 
$
R_j:=
\begin{bmatrix}
  1&0\\0&e^{2\pi i/2^j}
\end{bmatrix}
$
on the $k$th qubit 
controlled by the $m$th qubit, 
$S_{jk}$ is the SWAP gate $S$ on the $j$th and $k$th qubits, 
and $n_2:=\floor{n/2}$ (integer part of $n/2$). 
The number of gates of this construction is $n(n+1)/2+\floor{n/2}$ so that 
$G(U\psp n\qft)\leqslant n(n+1)/2+\floor{n/2}\approx n^2$.

The second example of the target $U_f$ is chosen so that 
we can expect that $U_f$ has the gate complexity $G(U_f)$ exponential
in $n$. 
To do so, we pick a $U_f$ which does not have any apparent symmetry
because 
a generic unitary operator $U_f$ 
has the gate complexity $G(U_f)\approx 4^n$ and 
a unitary operator $U_f$ with a generic image of a fixed state vector, 
$U_f\ket0$, 
has the gate complexity $G(U_f)\approx 2^n$~\cite{BulOLeBre05PRL}. 
Our concrete choice is $U\psp n\nos$ which is, in the matrix form,  
\begin{align}
  U\psp n\nos{}=
  \begin{bmatrix}
    \gamma_0 \alpha_0 & \gamma_1 \alpha_0& \gamma_2 \alpha_0 
    & \cdots & \cdots & \gamma_{n-1}\alpha_{0}\\
    \gamma_0 \alpha_1 & \gamma_1 \beta_1 & \gamma_2 \alpha_1 
    & \cdots & \cdots & \gamma_{n-1}\alpha_{1}\\
    \vdots & 0   & \gamma_2 \beta_2 
    & \cdots & \cdots & \gamma_{n-1}\alpha_{2}\\
    \vdots & \vdots & 0  & \ddots & &\vdots\\
    \vdots & \vdots & \vdots & &\ddots &\vdots\\
    \gamma_0 \alpha_{n-1} & 0 &0 &\cdots &\cdots 
    & \gamma_{n-1}\beta_{n-1}\\
  \end{bmatrix}, 
  \label{eq-Unos}
\end{align}
where $\alpha_k:=(k+1)^{1/3}e^{i\sqrt k}$, 
and $\beta_j\in\C$ 
and 
$\gamma_j>0$ are determined by 
orthogonalization and 
normalization, respectively, of the columns. 
The state vector $  U\psp n\nos\ket0$ is given by the first column of 
the right hand side of \Ref{eq-Unos} 
which does not have apparent symmetry. 
One can therefore expect that the gate complexity is exponential in $n$,  
$G(U\psp n\nos)\grsim 2^n$.

For each given time $T$, the convergence of the scheme 
is checked by the convergence of the fidelity, and 
the constancy of $\lambda(t)$ explained in Sec.~\ref{krotov}.

\begin{figure}[tbp]
  \includegraphics
  [width=1.02\linewidth]
  {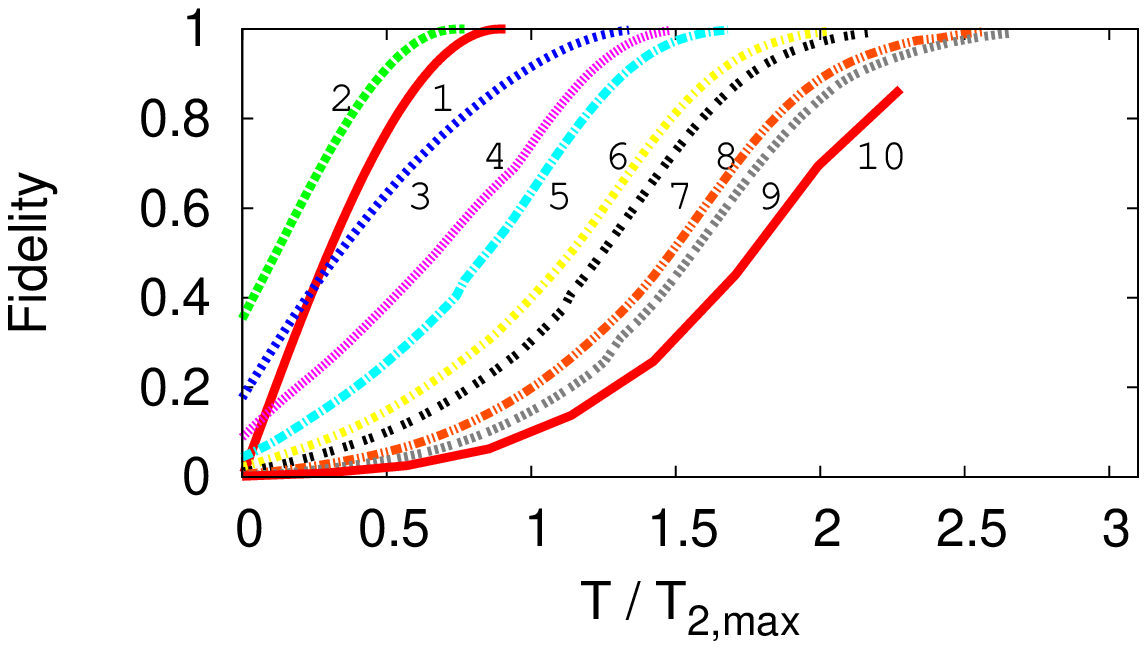}

  \vspace*{-1cm}
  \includegraphics
  [width=1\linewidth]
  {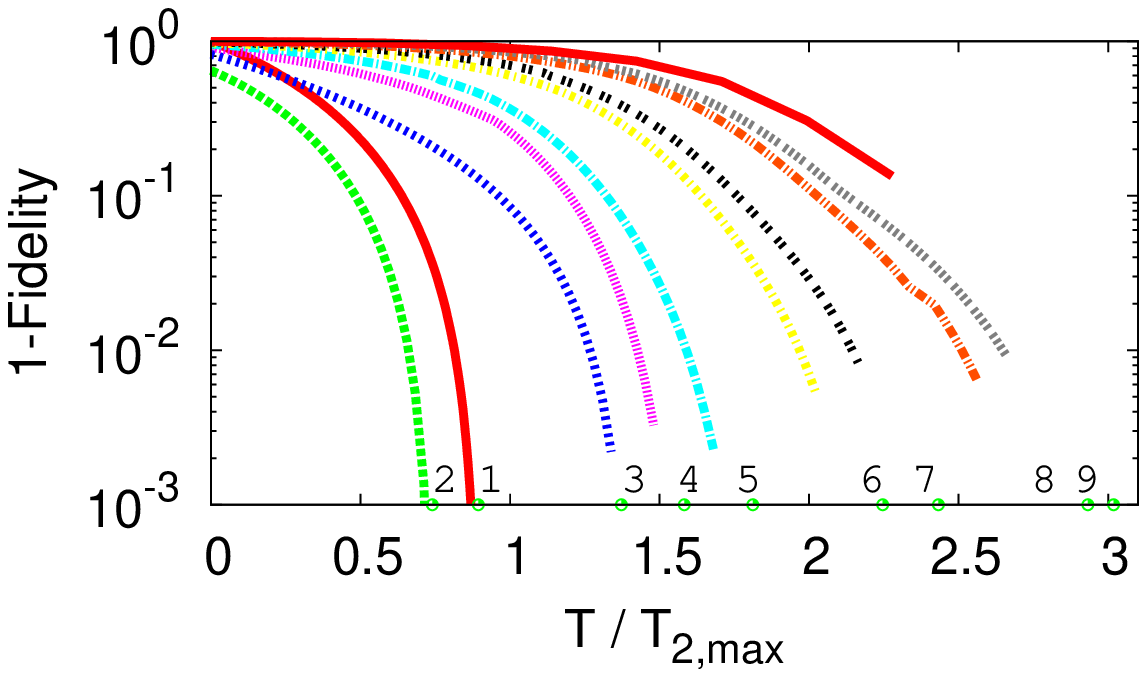}
  \caption{\prepri{(Color online) }
    (Upper) Optimal achievable fidelity in given
    physical time 
    for the $n$-qubit QFT 
    in {\fqct}. 
    The numbers beside the curves denote $n$. 
    The $n=10$ result is preliminary. 
    One can see property (a) in Sec.~\ref{f-t}. 
    (Lower) The same graph with the vertical axis being $1-\FF$ 
    in logarithmic scale. 
    The points and the numbers on the lower axis indicate 
    the time complexity
    $T(U_f)$ estimated by the fidelity$\to1$ limit for each $n$. 
    }
    \label{fig-f-vs-t}
\end{figure}

\subsection{Fidelity-time relation}
\label{f-t}
Let us show 
the fidelity-time relation, namely, 
the maximal achievable fidelity $\FF(U(T),U_f)$ in time 
$T$.
For simplicity, we only show the case of the QFT (Fig.~\ref{fig-f-vs-t}). 

Though our scheme in Sec.~\ref{krotov} is monotonic, 
there is a possibility of the calculation is 
trapped by a local minimum of the action $S$.  
We did the following two things 
in order to find 
the global maximum $\FF$. 
One is simply preparing many random seeds $h_a(t)$ for each $T$. 
Another is making use of the continuity of the solutions. 
Namely, we prepared many random seeds $h_a(t)$ for some fixed $T$. 
Then we used the solution for $T$ 
as the seed for a nearby $T$, 
and find continuous branches of locally optimal solutions. 
This ``output recycling'' turned out to be often more powerful 
than merely increasing the number of random seeds for every $T$. 
We observed crossovers of those branches. 
Only the branches of the largest $\FF$ contribute
to the curves in Fig.~\ref{fig-f-vs-t}.  

We observe the following, which 
may be characteristic of the QFT. 

(a) The odd and even qubits seem to make a pair (4-5, 6-7, and 8-9)
for $n\geqslant4$. 

In other words, 
the time-optimal solutions seems to split into the series of odd $n$ and
that of even $n$, which may be useful in the future mathematical 
analysis of the time-optimal solutions of the QFT.

\subsection{The limit $\FF\to1$}
We estimate $T(U)$ from the limit $\FF\to1$ of the solutions to {\fqct} 
in Sec.~\ref{f-t}. 

In the limit, we have two sources of error. 
One is a natural numerical error which makes the fidelity $\FF(U(T),U_f)$
saturate below unity. 
Another is that if $T$ is larger or very close to the time complexity $T(U_f)$, 
the solution of {\fqct} for the physical time $T$ begin to ``take a
roundabout route'' before reaching $U(T)$. 

With this behaviour in mind, 
we estimate the time complexity by a nonlinear
fitting of the fidelity-time curve around $\FF(U(T),U_f)\grsim0.99$. 
Fig.~\ref{fig-f=1-limit} shows an example of the estimation, in the case of the
$n=5$ QFT. 
We fit 
$y:=1-\FF(U(T),U_f)$ by 
$y=a(b-x)^{c}$, where 
$x:=T/\Tmax$. 
The estimation is given by $T(U_f)/\Tmax\simeq b$. 
The error in the time complexity $T(U)/\Tmax$ is about $0.1$ 
in the case of QFT, and 
it is about $0.1$ or $0.2$ 
in the case of 
asymmetric unitary operator. 
This does not change the conclusion of the subsequent sections.

\subsection{Time complexity as function of $n$}
\label{temporal}
Let us discuss the behaviour of the time complexity as 
a function of number of qubits, $n$. 

Fig.~\ref{fig-t-vs-n-QFT} shows 
the time complexity $T(U\psp n\qft)$ in {\tqct} 
as a function of $n$, 
where $T(U\psp n\qft)$ is obtained by the value of $T$ 
on each curve in Fig.~\ref{fig-f-vs-t} in the limit $\FF\upconv1$. 
For $n=1$ and $n=2$, 
one has the exact values 
$\frac{T(U\psp1\qft)}{\Tmax}
=\frac{\pi}{2}/ \frac{\sqrt5\pi}{4}
=\frac2{\sqrt5} \approx0.8944$ 
and 
$\frac{T(U\psp2\qft)}{\Tmax}
=\frac{\sqrt{11}\pi}{8}/ \frac{\sqrt5\pi}{4}
={\frac{\sqrt11}{2\sqrt5}}\approx0.7416$, 
calculated from \Ref{eq-T-for-2bit}, 
with which our numerical results agree well. 

\begin{figure}[t]
  \centering
  \includegraphics
  [width=.8\linewidth]
  {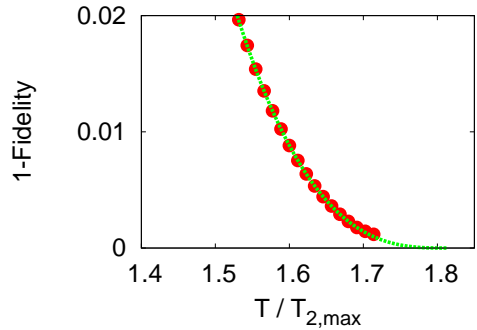}
  \caption{\prepri{(Color online) }
    Estimation of $T(U_f)$ 
    from the limit $\FF\to1$. 
    In the $n=5$ QFT example, 
    $y:=1-\FF(U(T),U_f)$ is fitted by 
    $y=0.735(1.81-x)^{2.84}$, where 
    $x:=T/\Tmax$, 
    giving an estimate 
    $T(U_f)/\Tmax\simeq 1.81$. 
    We used the data $0.002\leqslant y \leqslant 0.01$. 
    }
  \label{fig-f=1-limit}
\end{figure}

We observe the following property 
from Fig.~\ref{fig-t-vs-n-QFT}. 

(b)
The optimal time $T(U\psp n\qft)$ 
is linear in the number of qubits, $n$. 

The line in Fig.~\ref{fig-t-vs-n-QFT} is the result of a linear fitting, 
which is $T(U_f)/\Tmax=0.32n+0.27$. 
We used the data $n\geqslant2$ in the fitting because the 
behaviour of $n\leqslant2$ and that of $n\geqslant2$ should differ 
due to the nature of $\Gamma$ allowing only 
interactions involving two qubits or less.

Property (b) is in good contrast to the number of gates, 
$O(n^2)$, of the 
known efficient algorithm~\Ref{eq-fourier-algorithm} 
for the $n$-qubit QFT. 
However, it can be understood naturally. 
Since $\frac{T(W)}\Tmax=\frac{2}{\sqrt5}$, 
$\frac{T(S)}{\Tmax}=\sqrt{\frac35}$ 
and 
$\frac{T(R_{j,km})}{\Tmax}=\frac1{2^{j-1}}\sqrt{\frac35}$ from 
\Ref{eq-T-for-2bit}, 
we have the physical time cost $T'{}\psp n$ of the construction 
\Ref{eq-fourier-algorithm} 
is 
\cite{fn-c} 
\begin{align}
  &\frac{T'{}\psp n}{\Tmax}
    \nn
    \textstyle
  =&
  \frac1{\Tmax}
  \bracket{
    nT(W)
    +\sum_{j=2}^n(n-j+1)T(R_{j,km})
    +\left\lfloor \frac n2\right\rfloor T(S)
  }
  \nn
   =&
   \frac{2n}{\sqrt5}
   +
   \sqrt{\frac35}
   \sum_{j=1}^{n-1}\frac{n-j}{2^j}
   +{\sqrt\frac35}{\left\lfloor \frac n2\right\rfloor }
  \nn
   =&
   \frac{2n}{\sqrt5}
   +
   \sqrt{\frac35}
   \paren{
     n-2+\frac1{2^{n-1}}+
     \left\lfloor \frac n2\right\rfloor 
   }, 
   \label{eq-time-qft-alg}
\end{align}
which is bounded (from above and below)  by a linear function of $n$. 
The time complexity $T(U\psp n\qft)$ is several times smaller than  
$T'{}\psp n$ 
(except for $n=1$ when they coincide). 
The significance of 
$T'{}\psp n$ is that it 
is a rigorous upper bound for the time complesity, 
$T(U\psp n\qft)\leqslant T'{}\psp n$, 
which implies that 
$T(U\psp n\qft)$ {\em is at most linear in $n$}. 
This strongly supports property (b) 
and the correctness of the
numerical calculation.

Fig.~\ref{fig-t-vs-n-gene} shows the time complexity $T(U_f)$ as a function of
$n$ for the case of the asymmetric target $U\nos\psp n$. 

We observe that 

(b${}'$) The optimal time $T(U\psp n\nos)$ 
is exponential in the number of qubits, $n$. 

As in the case of QFT, we used the data $n\geqslant2$ to fit by a
function. 
The time complexity $T(U\nos\psp n)$ is well fitted by an
exponential function as  
$T(U\nos\psp n)/\Tmax=0.20\times 2^{0.82n}$. 
It is suggested from the numerical result
that 
$U\nos\psp n$ is in the class $T(U)\sim G(U)$. 

To conclude, it is suggested that 
both the QFT and 
the asymmetric unitary 
are in the class 
$T(U)\sim G(U)$, 
which is polynomial in $n$ for the former and 
exponential in $n$ for the latter. [Note that in the polynomial case, 
$T(U)\sim G(U)$ implies 
$T(U)\approx G(U)$].

\begin{figure}[t]
  \centering
  \includegraphics
  [width=\linewidth]
  {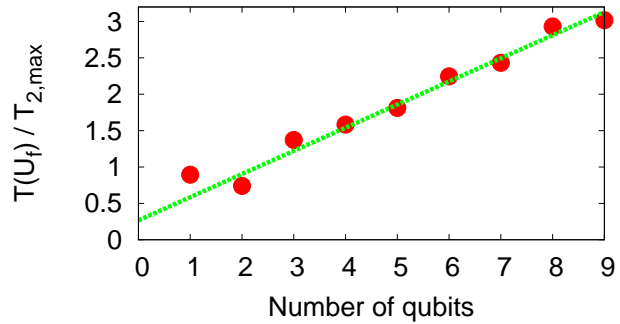}
  \caption{\prepri{(Color online) }
    Optimal time $T(U_f)$ for realization of 
    the QFT 
    as a function of the number of qubits, $n$, with a linear least squares
    fitting for $n\geqslant2$: 
    $T(U_f)/\Tmax=0.32n+0.27$. 
    }
    \label{fig-t-vs-n-QFT}
\end{figure}

\subsection{Behavior of the time-optimal Hamitonian}
Let us analyze the behavior of the optimal Hamiltonian $H(t)$ in {\tqct}. 

We shall say that an element of $\su(N)$ is 
{\em symmetric} (or {\em antisymmetric}) if it is so in the standard matrix 
representation. 
In particular, 
$\tau_a$ is symmetric (antisymmetric) 
if it contains even (odd) number of $\pa y$; for example, 
$\sigma^{12}_{yy}/\sqrt N$ 
is symmetric and 
$\sigma^{12}_{xy}/\sqrt N$ 
is antisymmetric. 

Fig.~\ref{fig-4ham} is the behavior of the 
Hamiltonian $H(t)$ for the 4-qubit QFT with $\FF\simeq1$, 
which can be considered as the solution of {\tqct}. 
The components $h_a$ are categorized into the four according to: 
whether $\tau_a$ is one- or two-qubit interaction, and 
whether $\tau_a$ is symmetric or antisymmetric. 

The results suggest the following. 

(c)
The components for symmetric (antisymmetric) 
$\tau_a$ is symmetric (antisymmetric) 
  under time-reversal $t\mapsto T-t$, 
and 

(d)
  one-qubit interaction components are constant. 
Note that (c) and (d) imply that one-qubit antisymmetric components 
vanish.

The same properties are safisfied by the $n=5$ QFT. 
For the $n=3$ QFT, (c) does not hold but (d) does.

For the asymmetric target $U\nos\psp n$, 
we do not observe property (c), 
the time reversal invariance. 
However, we do obverve property (d), constancy of one-qubit
components, also in these cases. 
Fig.~\ref{fig-MCN-3ham} shows the $n=3$ example 
of the behavior of the optimal Hamitonian. 

These properties 
will be discussed from a theoretical point of view in Sec.~\ref{math}.

\begin{figure}[t]
  \centering
  \includegraphics
  [width=1.01\linewidth]
  {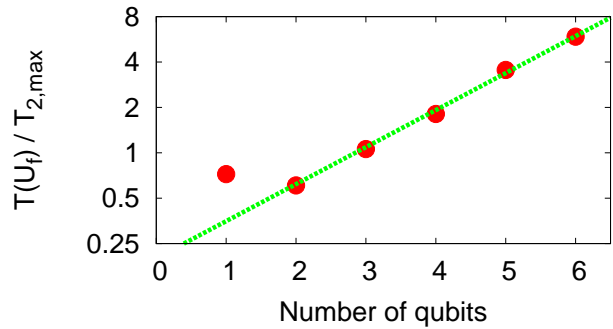}
  \caption{\prepri{(Color online) }
    Optimal time $T$ for realization of 
    the asymmetric unitary $U\nos\psp n$
    as a function of the number of qubits, $n$, with a linear least squares
    fitting for $n\geqslant2$: 
    $T(U_f)/\Tmax=0.20\times 2^{0.82n}$. 
    Note that, 
    in contrast to
    Fig~\ref{fig-t-vs-n-QFT}, 
    the vertical axis is in {\em logarithmic}\/ scale. 
    }
    \label{fig-t-vs-n-gene}
\end{figure}

\section{Mathematical justification of the behaviour of the time-optimal
  Hamiltonian} 
\label{math}

The temporal behaviour of the optimal Hamiltonian in {\tqct} found in 
Sec.~\ref{res}, 
property (c) for some of the QFT and property (d) for the QFT and 
the asymmetric unitary, 
are not peculiar to the case of those target unitary operators. 
They in fact can be proven under a fairly general condition. 
These results also serve as evidences of reliability of the numerical
calculation. 

\begin{figure}[t]
  \centering
  \includegraphics
  [width=\linewidth]
  {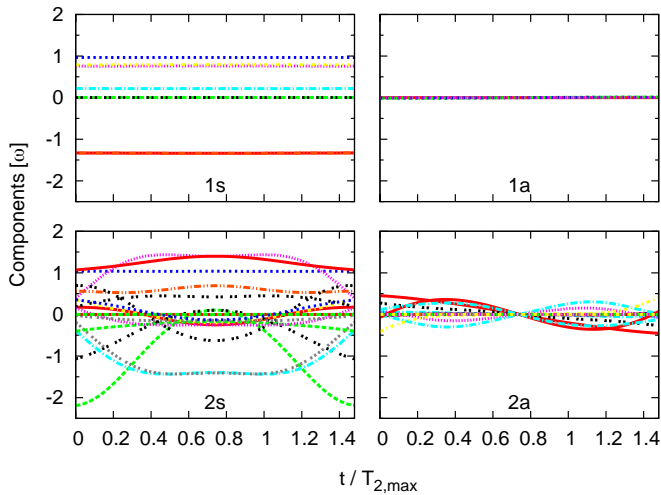}
  \caption{\prepri{(Color online) }
    Behavior of the optimal Hamiltonian for
    the 4-qubit QFT with 
    $T/\Tmax=1.48$, 
    $\FF=.9967$. 
      Shown are 
      the components with respect to the generators which are 
      (1s) one-qubit symmetric,  
      (1a) one-qubit antisymmetric, 
      (2s) two-qubit symmetric, and
      (2a) two-qubit antisymmetric, 
      in the sense of Sec.~\ref{temporal}. 
      The solution is time-reversal invariant. 
      The one-qubit components are constant in time. 
    }
\label{fig-4ham}
\end{figure}

\subsection{Time reversal invariance}
\label{math-time-reversal}

Let us show (c), the time-reversal symmetry 
found in some of the time-optimal solutions in Sec.~\ref{res}. 

Let us define the {\it time reversal}\/ 
$(\TR U(t),\TR H(t),\TR V(t),\TR \lambda(t))$ 
of the set of variables, $(U(t),H(t),V(t),\lambda(t))$, by 
\begin{align}
  \TR U(t)&:=U^*(T-t)U^T(T), 
  \nn
  \TR H(t)&:=H^*(T-t)={\textstyle \sum_a}h_a(T-t)\tau_a^*, 
  \nn
  \TR V(t)&:=V^*(T-t)U^T(T), 
  \nn
  \TR\lambda(t)&:=\lambda(T-t), 
\end{align}
where the superscript asterisk denotes the complex conjugate 
and 
the superscript $T$ denotes the transpose. 
We have the following, 
whose proof is given in Appendix~\ref{pf-time-reversal}. 
\begin{theorem}
\label{th-time-reversal}
Let the target $U_f$ be symmetric up to phase. 
If $(U(t),H(t),V(t),\lambda(t))$ is a solution of {\tqct} for $U_f$ with 
  optimal time $T$, 
  so is $(\TR U(t),\TR H(t),\TR V(t),\TR \lambda(t))$. 
  In particular, if the solution is unique, 
  it is invariant under time-reversal, 
  $(U(t),H(t),V(t),\lambda(t))
  =(\TR U(t),\TR H(t),\TR V(t),\TR \lambda(t))$.  
\end{theorem}

The QFT is 
a symmetric target.
The theorem justifies the observed symmetry property (c) 
of the QFT for $n\ne3$. 
Convergence of ramdomly chosen initial Hamiltonians to 
a single time-symmetric solution suggests that 
the optimal solution is unique in those cases. 
The 3-qubit case did not show the symmetry, and 
consistently we observed many optimal solutions $H(t)$.

The asymmetric target unitary is not a symmetric operator 
and 
the numerical solutions did not show time-reversal invariance. 

We remark that 
Theorem~\ref{th-time-reversal}
holds not only in the present version of {\tqct} 
but also in any quantum brachistochrone
with $\Gamma^*=\Gamma$.

\begin{figure}[tp]
  \centering
  \includegraphics
  [width=\linewidth]
  {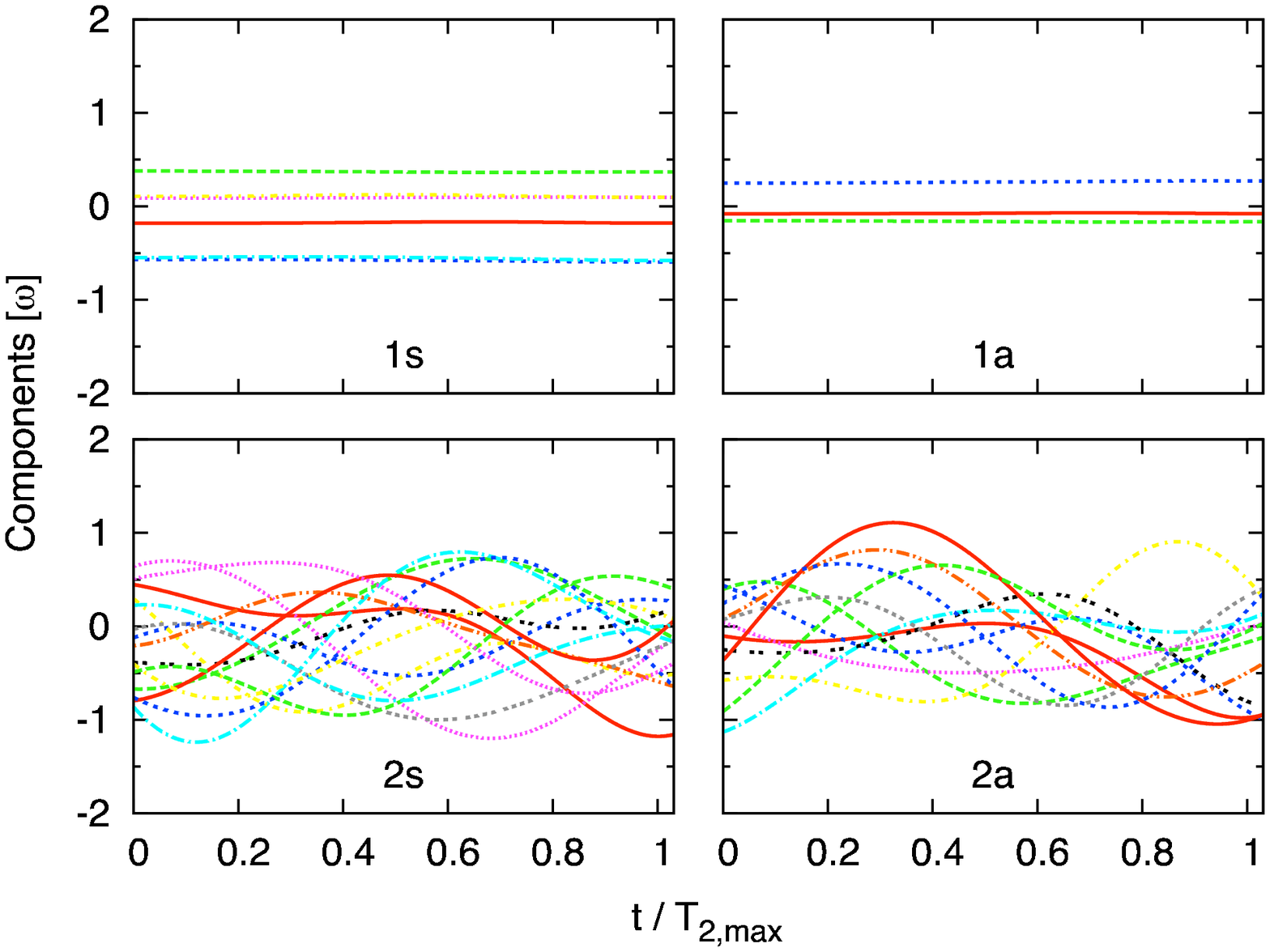}
  \caption{\prepri{(Color online) }
    Behavior of the optimal Hamiltonian for
    the 3-qubit 
    asymmetric unitary $U\psp n\nos$ 
    with 
    $T/\Tmax=1.03$, 
    $\FF=.9997$. 
      Shown are 
      the components with respect to the generators which are 
      (1s) one-qubit symmetric,  
      (1a) one-qubit antisymmetric, 
      (2s) two-qubit symmetric, and
      (2a) two-qubit antisymmetric, 
      in the sense of Sec.~\ref{temporal}. 
      The solution is {\em not}\/ time-reversal invariant. 
      The one-qubit components are constant in time. 
    }
\label{fig-MCN-3ham}
\end{figure}

\subsection{Constancy of one-qubit components}
\label{math-constancy-1bit}

Let us prove (d), which turns out to hold in general. 
\begin{theorem}
The one-qubit part of the Hamiltonian for any solution of {\tqct} 
is constant in time.
\end{theorem}
\begin{proof}
Let 
$\g_j$ be the space of $j$-qubit
operations. 
They satisfy the following commutation relations~\cite[Sec.~V] {paper2}: 
\begin{align}
  \label{eq-g-comm}
  &[\g_j,\g_k]=\g_{|j-k|+1}\oplus\g_{|j-k|+3}
  \oplus\cdots\oplus\g_{j+k-1}, 
\end{align}
where $j,k\geqslant1$ and $\g_j:=0$ for $j>n$. 
From \Ref{eq-lambda}, we can decompose $F$ as 
\begin{align}
F= \lambda H+F'
\end{align}
with 
  $H=\sum_{j=1,2}H_j$ 
  and
  $F'=\sum_{j\ge3}F_j$, 
where an integer subscript $j$ denotes the projection to $\g_j$. 
Then it follows from 
\Ref{eq-g-comm} and $\lambda=$ const. 
that 
the equation $i\dot F=[H,F]$ 
can be written as 
\begin{align}
  i\dot H_1&=0, 
  \nn
  i\lambda\dot H_2&=[H_2,F_3]_2, 
  \nn
  i\dot F_3&=[H_1,F_3]+[H_2,F_4]_3, \nn
  &.... 
  \label{eq-h1} 
\end{align}
This implies $H_1=$ const. for any target unitary
$U_f$. 
\end{proof}

Since constancy of one-qubit components is a quite general feature, 
it is a useful criterion of the convergence of the numerical scheme.

\section{Conclusion}
\label{conc}
We investigated the simplest version of 
{\tqct}, 
where 
the computation time 
is defined by 
the physical one 
and 
the Hamiltonian contains only one- and two-qubit interactions. 
This version of {\tqct} 
is also considered as optimality by sub-Riemannian geodesic length.

Motivated by the relations between time complexity and gate 
complexity~\Ref{eq-T-vs-G}, 
we aimed to pursue the possibility of using time complexity as a 
tool to estimate gate complexity, 
and asked the following question: to what extent is true the statement 
that 
{\em time complexity is polynomial in the number of qubits if and only if  
so is gate complexity.}
In particular, we want to identify the classes of unitary operators $U$ 
that satisfy $T(U)\approx G(U)$ and $T(U)\sim G(U)$, by which we meant 
$T(U)$ and $G(U)$ are bounded by polynomial of each other and $n$, 
and $T(U)$ is polynomial in $n$ if and only if so is $G(U)$, 
respectively. 

For this program, 
we introduced an efficient Krotov-like numerical scheme 
by making use of the relation between {\tqct} and {\fqct} and 
the formal similarity of the latter to OCT, 
and showed its monotonic convergence property. 

We chose the quantum Fourier transform as an example of the target with 
polynomial $G(U)$ and 
a unitary operator without symmetry that is expected to have 
exponential gate complexity. 
We obtained the fidelity-time relation, time copmlexity $T(U)$, 
The time complexity of the QFT is found to be linear
in the number of qubits. 
The time complexity of 
the target without symmetry 
is exponential in $n$. 
These results suggest that the QFT and 
the asymmetric target 
are both 
in the class $T(U)\approx G(U)$, 
and that $T(U)\approx G(U)$ is linear in $n$ for the QFT and is exponential in
$n$ for C${}^{n-1}$-NOT, respectively. 
This supports the usefulness of time complexity as a tool to estimate gate
complexity. 
It is also suggested that 
a polynomial-gate algorithm does not exist indeed 
for the asymmetric target $U\nos\psp n$, 
because $T(U)/\Tmax$ is the absolute lower bound of $G(U)$.

We also found two characteristics of the optimal Hamiltonian $H(t)$ 
of {\tqct}. 
One is symmetry under time reversal and the other is 
constancy of one-qubit operation, 
which are mathematically shown to hold 
in fairly general situations.

A natural extension of this work is to push forward with 
the program above 
by comparing the time complexity (or equivalently, the arc-length of the
sub-Riemannian geodesics) 
and the gate complexity for other unitary operators. 
Another direction is to consider other variants of {\tqct}. 
An example is {\tqct}
which allows only nearest neighbor interactions in a lattice, 
and 
another is {\tqct}
where only the time spent in two-qubit interactions is counted and 
that spent in one-qubit interactions is neglected. 

We would like to stress that
interplay of numerical and mathematical methods, of which 
an example was the argument given in Sec.~\ref{math}, 
is important in analyzing {\tqct}. 
We hope that our method will lead to 
new understanding about the power of quantum computation.

\section*{Acknowledgments}
We sincerely thank Professor Akio Hosoya 
and Professor Alberto Carlini for
fruitful discussions.

\appendix
\section{Monotonicity of our Krotov-like scheme}
\label{krotov-mono}

Let us show 
the monotonicity of the Krotov-like scheme given in Sec.~\ref{krotov}. 

Let $\Delta S$ be the change of $S$ in the cycle (iii)--(iv). 
We would like to show $\Delta S\leqslant 0$. 
Let the variables 
$(U(t),h(t),V(t),\lambda(t))$ be
\begin{align*}
(U\old(t),h\old(t),V\old(t),\lambda\old(t))  
\end{align*}
after step (iv) which will be the 
inputs to a new cycle (iii)--(iv). 
They satisfy \Ref{eq-schu}, \Ref{eq-iniu}, 
\Ref{eq-normalization}, 
\Ref{eq-lambda}, 
but not \Ref{eq-schv} or \Ref{eq-finv}. 
Let the variables be 
\begin{align*}
(U\old(t),\wb h(t),V\new(t),\wb\lambda(t)) 
\end{align*}
with 
$V\new=V\old+\wb\delta V$ 
and $\wb h=h\old+\wb\delta h$ 
after step (iii). 
They satisfy \Ref{eq-normalization}, 
\Ref{eq-schv}, 
\Ref{eq-lambda}, 
\Ref{eq-finv}, 
but not \Ref{eq-schu} or \Ref{eq-iniu}. 
Let the variables be 
\begin{align*}
(U\new(t),h\new(t),V\new(t),\lambda\new(t)) 
\end{align*}
with 
$U\new=U\old+\delta U$ and 
$h\new=\wb h+\delta h$
after step (iv). 
They satisfy \Ref{eq-schu}, \Ref{eq-iniu}, 
\Ref{eq-normalization}, 
\Ref{eq-lambda}, 
but not \Ref{eq-schv} or \Ref{eq-finv}. 

The change $\Delta S$ in the action \Ref{eq-action} 
after one cycle (iii)--(iv) is given by 
\begin{align}
  N\Delta S
  =&-N^2\FF(U\new(T),U_f)^2+N^2\FF(U\old(T),U_f)^2
  \nn
  =&-2\Re \Tr \da U\old(T) U_f \Tr \da U_f \delta U(T) 
  -|\Tr \da U_f \delta U(T) |^2, 
  \label{eq-monotonicity-1}
\end{align}
because both $(U\old,h\old)$ and $(U\new,h\new)$ 
satisfy \Ref{eq-schu} and \Ref{eq-normalization} 
and 
make the integrand in \Ref{eq-action} vanish. 
The first term on the RHS of \Ref{eq-monotonicity-1} is 
\begin{align}
  &-{2}\Re\Tr i\da V\new(T)\delta U(T)
  =-\int_0^Tdt\; {2}\Re\Tr i(\da V\new \delta U)^\bullet
  \nn
  &=
  -\int_0^Tdt (
  \lambda\new h\new\cdot\delta h
  +\wb\lambda\, \wb h\cdot\wb\delta h)
  \nn
  &=
  -\frac12\int_0^Tdt (
  \lambda\new |\delta h|^2
  +\wb\lambda |\wb\delta h|^2)\leqslant0, 
  \label{eq-monotonicity-2} 
\end{align}
where the dot denotes the inner product
and we have used 
\begin{align}
V\new(T)&=\frac iN U_f{\Tr U_f^\dagger U\old(T)},  \nn
\delta U(0)&=0, \nn
i\delta\dot U&=\delta H U\new+\wb\delta H U\old+\wb H\delta U, \nn
i\dot V\new&=\wb HV\new, \nn
\lambda\new h\newup_a&=2\Re \Tr \tau_aU\new\da V\new,
\end{align}
which follow from the conditions 
satisfied by each of the variables, given in the previous paragraph. 
One can see 
the last equality in 
\Ref{eq-monotonicity-2} 
by noticing that $(h\new+\wb h)/2$ is orthogonal to $\delta h$ because 
$|h\new|=|\wb h|$, and so forth. 

We conclude that our scheme is monotonic 
because \Ref{eq-monotonicity-1} and \Ref{eq-monotonicity-2} imply 
$\Delta S\leqslant0$. 
Note that 
$\delta U$ etc. above have not been assumed 
to be small.

\section{Proof of Theorem~\ref{th-time-reversal}}
\label{pf-time-reversal}
It is easily verified that if 
$(U(t),H(t),V(t))$ is a solution of 
\Ref{eq-schu}--\Ref{eq-normalization} 
and 
\Ref{eq-schv}--\Ref{eq-lambda}, 
so is 
$(\TR U(t),\TR H(t),\TR V(t))$. 
Eq. \Ref{eq-schu} is seen by 
\begin{align}
&i \dot U\TRB (t)
=-i\dot U^*(T-t)U^T(T)
\nn
&=H^*(T-t)U^*(T-t)U^T(T)
=\TR H(t)U\TRB (t). 
\end{align}
Eq. \Ref{eq-iniu} follows from 
$\TR U(0)=U^*(T)U^T(T)=\id$; 
Eq. \Ref{eq-finu} follows from 
\begin{align}
\TR U(T)=U^*(0)U^T(T)
=e^{-i\chi}U_f^T=e^{-i\chi}U_f. 
\end{align}
Eqs. \Ref{eq-Gamma} and \Ref{eq-normalization} 
follow from $\tau_a^*=\pm\tau_a$. 
Eq. \Ref{eq-lambda} is 
equivalent to 
\begin{align}
\lambda(t)H(t)
=\sum_a\tau_a\Tr\tau_a 
\paren{U(t)\da V(t)+V(t)\da U(t)}, 
\label{eq-lambda-3} 
\end{align}
The time reversal satisfies the same equation 
\Ref{eq-lambda-3} because 
\begin{align}
&\TR\lambda(t) \TR H(t)
=\lambda(T-t) H^*(T-t)
\nn
&=\bracket{\sum_a\tau_a\Tr\tau_a 
\paren{U(T-t)\da V(T-t)+V(T-t)\da U(T-t)}}^* 
\nn
&=\sum_a\tau_a\Tr\tau_a 
\paren{\TR U(t)\da V\TRB(t)+\TR V(t)\da U\TRB(t)}, 
\end{align}
where we have used $\tau_a^*=\pm\tau_a$ (the signature depends 
on $a$). 
Therefore the time reversal of the solution of {\tqct} with 
the target $U_f$ and the time $T$ is a solution of the same problem.


\begin{thebibliography}{99}
\bibitem{fn-a} 
Precisely speaking, 
for the variational principle to be properly defined, we must use the action 
which has a fixed upper limit in the integral. 
Here, for simplicity, 
we present the action \Ref{eq-action} which does not. 
This is justified, however, because 
the simple variations thereof lead to the same 
Euler-Lagrange equations thanks to the 
time reparametrization invariance of $L_T$ and the action. 
See also the remarks on the action in Ref.~\cite{paper3}.
\bibitem{fn-b} 
Constancy of $\lambda(t)$ is shown in Ref.~\cite{paper2}, 
but here we give a brief proof. 
On one hand, we have 
$\Tr \dot FH
=-i\Tr{ [H,F]H}=0$
because of 
\Ref{eq-brach}. 
On the other hand, we have 
$\Tr \dot FH
=\sum_a(\lambda h_a)^\bullet h_a
=\dot \lambda\sum_a h_a^2
+\lambda\paren{\sum_a h_a h_a}^\bullet/2
=N\omega^2\dot\lambda$, 
where 
the first equality follows from 
\Ref{eq-Gamma} and
\Ref{eq-lambda}, 
and the last from 
\Ref{eq-normalization}. 
These imply $\dot\lambda(t)=0$.
\bibitem{fn-c}
In the calculation of \Ref{eq-time-qft-alg}, we have used 
$\sum_{j=1}^k jr^j=\frac{r}{1-r}\paren{\frac{1-r^k}{1-r}-kr^k}$
 so that $\sum_{j=1}^{n-1} \frac{j}{2^j}=2-\frac{n+1}{2^{n-1}}$.
\end{thebibliography}
\end{document}